\def\arxiv{1} % 1 iff arXiv version                     %%
\ifnum\arxiv=1
\documentclass[12pt]{article}
\usepackage[a4paper,hmargin=1in,vmargin=1.1in]{geometry}
\usepackage{amssymb,amsmath,amsthm}
\usepackage[UKenglish]{babel}
\else%
\documentclass[a4paper,UKenglish]{lipics-v2016}
\usepackage{microtype}%if unwanted, comment out or use option "draft"
\fi

\usepackage{amsmath,amssymb}
\usepackage{hyperref}
\usepackage[ruled,vlined,linesnumbered]{algorithm2e}
\usepackage{paralist}
\usepackage{graphicx}
\graphicspath{{./figures/}}

\usepackage{verbatim}
\ifnum\arxiv=1
\newtheorem{theorem}{Theorem}
\newtheorem{corollary}[theorem]{Corollary}
\newtheorem{lemma}[theorem]{Lemma}
\newtheorem{definition}{Definition}
\fi

\usepackage{amsthm,amsfonts}

\newtheorem{claim}[theorem]{Claim}

\usepackage{tabularx, calc}

\usepackage{color}
\usepackage{soul}

\newcommand{\moti}[1]{\hl{\textbf{M}: #1}}

\renewcommand{\deg}{d}
\newcommand{\congest}{{\sf CONGEST}}
\newcommand{\msg}[2]{M_{{#1}\rightarrow{#2}}}
\newcommand{\eps}{\varepsilon}
\newcommand{\prop}{{\cal P}}

\newcommand{\Prr}[1]{{\mathrm Pr}\left[ #1 \right]}

\DeclareMathOperator{\load}{load}
\DeclareMathOperator{\diam}{diam}

\DeclareMathOperator{\poly}{Poly}

\title{\textbf{Faster and Simpler Distributed Algorithms for Testing and Correcting Graph Properties in the \congest-Model}}

\if\arxiv=0
\titlerunning{Faster \& Simpler Dist. Algs. for Testing \& Correcting Graph Properties in CONGEST}

\author[1]{Guy Even}
\author[2]{Reut Levi}
\author[3]{Moti Medina}
\affil[1]{%School of Electrical Engineering\\
 Tel-Aviv University\\
  \texttt{guy@eng.tau.ac.il}}
\affil[2]{Max Planck Institute
	for Informatics, Saarland Informatics Campus, Germany\\
\texttt{rlevi@mpi-inf.mpg.de}}
\affil[3]{Max Planck Institute
	for Informatics, Saarland Informatics Campus, Germany\\
\texttt{mmedina@mpi-inf.mpg.de}}

\authorrunning{Even, Levi, and Medina}
\Copyright{Guy Even, Reut Levi, and Moti Medina}
% mandatory: Please choose ACM 1998 classifications from http://www.acm.org/about/class/ccs98-html . E.g., cite as "F.1.1 Models of Computation".

\else
\author{
 	Guy~Even\thanks{Tel-Aviv University, Israel
 	\protect\url{guy@eng.tau.ac.il}}
 	\and
 	Reut~Levi\thanks{
 	Max Planck Institute
	for Informatics, Saarland Informatics Campus, Germany
 	\protect\url{{rlevi,mmedina}@mpi-inf.mpg.de}}
    \and 
    Moti~Medina$^\dagger$
}
\fi

\date{}

\begin{document}

\maketitle
%%%%%%%%%%%%%%
\begin{abstract}
  In this paper we present distributed testing algorithms of graph
  properties in the \congest-model~\cite{censor2016fast}.  Concretely,
  a distributed one-sided error $\eps$-tester for a property
  $\mathcal{P}$ meets the following specification: if the network has
  the property $\mathcal{P}$, then all of the network's processors
  should output YES, and if the network is $\eps$-far from having the
  property, according to a predetermined distance measure, then at
  least one processor outputs NO with probability at least $2/3$.

  We present one-sided error testing algorithms in the general graph
  model.

  We first describe a general procedure for converting $\eps$-testers
  with a number of rounds $f(D)$, where $D$ denotes the diameter of the
  graph, to $O((\log n)/\eps)+f((\log n)/\eps)$ rounds, where $n$ is
  the number of processors of the network. We then apply this
  procedure to obtain an optimal tester, in terms of $n$, for testing
  bipartiteness, whose round complexity is $O(\eps^{-1}\log n)$, which
  improves over the $\poly(\eps^{-1} \log n)$-round algorithm by
  Censor-Hillel et al. (DISC 2016). Moreover, for cycle-freeness, we obtain a
  \emph{corrector} of the graph that locally corrects the
  graph so that the corrected graph is acyclic. Note that,
  unlike a tester, a corrector needs to mend the graph in many places
  in the case that the graph is far from having the property.

  In the second part of the paper we design algorithms for testing
  whether the network is $H$-free for any connected $H$ of size up to
  four with round complexity of $O(\eps^{-1})$. This improves over the
  $O(\eps^{-2})$-round algorithms for testing triangle freeness by
  Censor-Hillel et al. (DISC 2016) and for testing excluded graphs of
  size $4$ by Fraigniaud et al. (DISC 2016).

  In the last part we generalize the global tester by Iwama and
  Yoshida~\cite{IwamaY14} of testing $k$-path freeness to testing the
  exclusion of any tree of order $k$. We then show how to simulate
  this algorithm in the \congest-model in $O(k^{k^2+1}\cdot\eps^{-k})$
  rounds.
\end{abstract}

\if\arxiv=0
\else
\fi

%\enddocument
%%%%%%%%%%%%%%%%%%%%%%%%%%%%%%%%%%%%
\section{Introduction}
In graph property
testing~\cite{goldreich1998property,goldreich2002property} the goal is
to design a sequential sublinear algorithm that, given a query access
to a graph, decides whether the graph has a given property or is
$\eps$-far from having it. By sublinear we mean that the number of
queries that the algorithm generates is much smaller than the size of
the graph. In the \emph{general graph model}, a graph, $G=(V,E)$, is $\eps$-far
from satisfying a property if at least $\eps \cdot |E|$ edges should be
added or removed so that the graph will have the property.
Graph property testing in the distributed \congest-model has been initiated recently
by Censor-Hillel et al.~\cite{censor2016fast}. In this setting, each processor
locally gathers in a parallel and synchronized fashion information
from the network while abiding to a logarithmic bandwidth
constraint. When the distributed algorithm terminates, each processor
outputs ACCEPT in case the graph (which acts as the network on which
the processors communicate) has the tested property (for one-sided error testing), and in case the
graph is $\eps$-far from the property, then there is at least one processor that
outputs REJECT with probability at least $2/3$. In the distributed
setting the goal is to design algorithms with small number of rounds.

In this paper we present improved testers to various properties, such
as: is the graph bipartite? Acyclic? Does the graph contain a copy of
predetermined tree of size $k$? or other subgraphs of size at most
$4$?  We obtain these results in three different ways: (1)~we
demonstrate a procedure that given an $\eps$-tester in the \congest-model with linear dependency on the network's diameter reduces it to a
logarithmic dependency in the size of the graph. This technique yields
improved algorithm for testing bipartiteness. (2)~We directly design
improved algorithms for testing whether the graph does not contain a
copy of a subgraph of size at most $4$. An important ingredient of
testing whether a copy of subgraph of size $4$ exists is the ability
to pick u.a.r. a path of length $2$ that emanates from a specific
vertex in the \congest-model. (3)~In the last part we design an
$\eps$-tester for the property of being free from copies of a specific
tree of size $k$. We then simulate this $\eps$-tester with the same
number of rounds as the number of queries of the sequential tester.

\subsection{Related Work}
Property testing in the distributed \congest-model was initiated by
Censor-Hillel et al.~\cite{censor2016fast}. In~\cite{censor2016fast},
the distributed testing model was defined as well as various testing
algorithms of whether a graph is: triangle-free, cycle-free, or
bipartite (i.e., free of odd cycles). Additionally, a logarithmic lower bound was proven for
testing bipartiteness and cycle-freeness. Finally, a simulation of
sequential (or global) testers for a certain class of graph properties
in the dense model is given. This simulation incurs a quadratic
blow-up w.r.t. the tester's number of queries.

Fraigniaud et al.~\cite{fraigniaud2016distributed} studied testing of
excluding subgraphs in the \congest-model. In~\cite{fraigniaud2016distributed} an algorithm for testing
whether a graph does not contain a specific subgraph (or any
isomorphic copy of it) of size $4$. The number of rounds of this
algorithm does not depend on the size of the graph. Fraigniaud et
al.~\cite{fraigniaud2016distributed} also consider the problem of
testing whether a graph excludes subgraphs of size $k\geq 5$ (e.g.,
$C_5$ or $K_5$).  For these properties, they present a ``hard'' family
of graphs for which some ``natural'' property testing algorithms have
a round complexity that depends on the size of the graph.

Our notion of correction is inspired by work on local
reconstruction of graph properties (see for example~\cite{CGR13}
and~\cite{KPS13}).  However, we note that our definition of correction
does not assume anything on the input graph. In particular, we do not
assume that it is close to having the required property.

\begin{comment}
  \moti{Yoshida?}Iwama and Yoshida~\cite{IwamaY14} \moti{Paritioning
    in CONGEST?}Elkin and Neiman~\cite{elkin2017efficient}
  \moti{Correctors?}\moti{testing related work?}
\end{comment}

\subsection{Our Contributions}
We design and analyze distributed testers in the distributed \congest-model all of which work in the general graph model. In this model, a graph $G=(V,E)$ is called $\eps$-far from having a property $\mathcal{P}$, if one must remove or add at least $\eps |E|$ edges from $G$ in order to obtain the property $\mathcal{P}$.

\subparagraph*{Diameter dependency reduction and its Applications.}
In Section~\ref{sec:compiler} we describe a general procedure for
converting $\eps$-testers with $f(D)$ rounds, where $D$ denotes the
diameter of the graph to a $O((\log n)/\eps)+f((\log n)/\eps)$ rounds,
where $n$ is the number of processors of the network. We then apply
this procedure to obtain an $\eps$-tester for testing whether a graph
is bipartite. The improvement of this tester over state of the art is
twofold: (a)~the round complexity is $O(\eps^{-1}\log n)$, which
improves over the $\poly(\eps^{-1} \log n)$-round algorithm by
Censor-Hillel et al.~\cite[Thm.~5.2]{censor2016fast}, and (b)~our
tester works in the general model while~\cite{censor2016fast} works in
the more restrictive bounded degree model. Moreover, the number of
rounds of our bipartiteness tester meets the $\Omega(\log n)$ lower
bound by~\cite[Thm.~7.3]{censor2016fast}, hence our tester is
asymptotically optimal in terms of $n$. We then apply this
``compiler'' to obtain a cycle-free tester with number of rounds of
$O(\eps^{-1}\cdot\log n)$, thus revisiting the result
by~\cite[Thm.~6.3]{censor2016fast}. The last application that we
consider is ``how to obtain a \emph{corrector} of the graph by using
this machinery?'' Namely, how to produce a an algorithm that locally
corrects the graph so that the corrected graph satisfies the
property. For cycle-freeness, we are able to obtain also a corrector.
Note that, unlike a tester, a corrector needs to mend
the graph in many places in the case that the graph is far from having
the property.

\subparagraph*{Testers for $H$-freeness for $|V(H)| \leq 4$.}  In
Section~\ref{sec:hfour} we design algorithms for testing (in the
general graph model) whether the network is $H$-free for any connected $H$
of size up to four with round complexity of $O(\eps^{-1})$. By
$H$-free we mean that there is no sub-graph $H'$ of $G$ such that $H$
is isomorphic to $H$. This improves over the $O(\eps^{-2})$-round
algorithms for testing triangle freeness by Censor-Hillel et
al.~\cite[Thm.~4.1]{censor2016fast} and for testing excluded graphs of
size $4$ by Fraigniaud et
al.~\cite[Thm.~1]{fraigniaud2016distributed}.

\subparagraph*{Testers for tree-freeness.}  In
Section~\ref{sec:bigtrees} we first generalize the global tester by
Iwama and Yoshida~\cite{IwamaY14} of testing $k$-path freeness to
testing the exclusion of any tree, $T$, of order $k$. Note that in
this part of the paper we do not make any assumption of the size of
the tree, e.g., $k$ can be larger than $4$. Our tester has a one sided
error and it works in the general graph model with random edge
queries. We, then, show how to simulate this algorithm in the \congest-model in $O(k^{k^2+1}\cdot\eps^{-k})$ rounds.

%\paragraph*{Paper Organization}
%%% Local Variables:
%%% mode: latex
%%% TeX-master: "distTesting"
%%% End:

%%%%%%%%%%%%%%%%%%%%%%%%%%%%%%%%%%%%
\section{Computational Models}
\subparagraph*{Notations.}  Let $G=(V,E)$ denote a graph, were $V$ is
the set of vertices and $E$ is the set of edges. Let $n \triangleq
|V(G)|$, and let $m \triangleq |E(G)|$. For every $v \in V$, let
$N_G(v) \triangleq \{u \in V \mid \{u,v\} \in E\}$ denote the
neighborhood of $v$ in $G$. For every $v \in V$, let $d_G(v)
\triangleq |N_G(v)|$ denote the degree of $v$. When the graph at hand
is clear from the context we omit the subscript $G$.

\subsection{Distributed \congest-Model}
Computation in the distributed \congest-model~\cite{peleg2000distributed} is done as follows.  Let $G=(V,E)$ denote a network where each
vertex is a processor and each edge is a communication link between
its two endpoints.  Each processor is input a local input.  Each
processor $v$ has a distinct ID - for brevity we say that the ID of
processor $v$ is simply $v$.\footnote{In this paper we focus on
  randomized algorithms, hence one can omit the assumption that each
  processor has a distinct ID.}  The computation is synchronized and
is measured in terms of \emph{rounds}. In each round (1)~each
processor does a local computation, and then (2)~sends (different)
messages of $O(\log n)$ bits to each of its neighbors (or a possible
``empty message''). After the last round all the processors stop and
output a local output.

\subsection{(Global) Testing Model}
Graph property
testing~\cite{goldreich1998property,goldreich2002property} is defined
as follows.  Let $G=(V,E)$ denote a graph. We assume that for each $v
\in V$ there is an arbitrary order on its set of neighbors.  Let $\cal
P$ denote a graph property, e.g., the graph is cycle-free, the graph
is bipartite, etc.  We say that a graph $G$ is \emph{$\eps$-far (in
  the general graph model)} from having the property $\cal P$ if at least
$\eps \cdot m$ edges $E'$ should be added or removed from $E(G)$ so
as to obtain the property ${\cal P}$.

An algorithm in this model is given a query access to $G$ of the form:
(1)~what is the degree of $v$ for $v\in V$? (2)~who is $i$th neighbor
of $v \in V$?

We say that an algorithm is a \emph{$\eps$-tester for property ${\cal
    P}$, one sided, in the general graph model} if given query access to the
graph $G$ the algorithm ACCEPTS the graph $G$ if $G$ has the property
$\cal P$, i.e, completeness, and REJECTS the graph $G$ with
probability at least $2/3$ if $G$ is $\eps$-far from having the
property $\cal P$, i.e., soundness.

The complexity measure of this model is the number of queries made to
$G$. Usually, the goal is to design an $\eps$-tester with $o(n)$
number of queries.

In Section~\ref{sec:bigtrees} an additional query type is allowed of
\emph{random edge query} where an edge $e$ is picked u.a.r. from $E$.

\subsection{Distributed Testing in the \congest-model}
Let $G=(V,E)$ be a graph and let ${\cal P}$ denote a graph property.
We say that a randomized distributed \congest\ algorithm is an \emph{$\eps$-tester
  for property ${\cal P}$ in the general graph model}~\cite{censor2016fast}
if when $G$ has the property ${\cal P}$ then \emph{all} the processors
$v \in V$ output ACCEPT, and if $G$ is $\eps$-far from having the
property ${\cal P}$, then there is a processor $v \in V$ that outputs
REJECT with probability at least $2/3$.

\subsection{Distributed Correcting}
In this section we define correction in the distributed setting. We
then explain how to obtain correction for the property of
cycle-freeness.

\begin{definition}
  A graph property $\prop$ is \emph{edge-monotone} if $G\in \prop$ and if $G'$ is
  obtained from $G$ by the removal of edges, then $G'\in \prop$.
\end{definition}

\begin{definition}\label{def:corr}
In the distributed \congest-model, we say that an algorithm is an $\eps$-corrector for an edge-monotone property $\mathcal{P}$ if the following holds.
\begin{enumerate}
\item Let $G = (V,E)$ denote the network's graph. When the algorithm terminates, each processor $v$ knows which edges in $E$, that incident to $v$, are in the set of {\em deleted edges} $E' \subseteq E$.
\item $G(V, E\setminus E')$ is in $\mathcal{P}$.
\item $|E'| \leq dist(G, \mathcal{P}) + \eps |E|$, where $dist(G, \mathcal{P})$ denotes the minimum number of edges that should be removed from $G$ in order to obtain the property $\mathcal{P}$.
\end{enumerate}
\end{definition}

%\begin{definition}
%	$A$ is an $(F,P)$-corrector if
%	\begin{inparaenum}[(i)]
%		\item $\outt(A) = F$,
%		\item $\dist(G \setminus F,P)\leq \epsilon$,
%		\item $|F|\leq \dist(G,P)$.
%	\end{inparaenum}
%\end{definition}

%%% Local Variables:
%%% mode: latex
%%% TeX-master: "distTesting"
%%% End:

%%%%%%%%%%%%%%%%%%%%%%%%%%%%%%%%%%%%
\section{Reducing the Dependency on the Diameter and Applications}\label{sec:compiler}
In this section we present a general technique that reduces the
dependency of the round complexity on the diameter. The technique is
based on graph decompositions defined below.
\begin{definition}[\cite{miller2013parallel}]
 Let $G=(V,E)$ denote an undirected graph.
A $(\beta,d)$-decomposition of $G$ is a partition of $V$ into disjoint subsets $V_1,\ldots,V_k$ such that
 \begin{inparaenum}[(i)]
 \item For all $1\leq i \leq \ell$, $\diam(G[V_i]) \leq d$, where
   $G[V_i]$ is the vertex induced subgraph of $G$ that is induced
   by $V_i$.
 \item The number of edges with endpoints belonging to different subsets is at most $\eps \cdot |E|$.
We refer to these as \emph{cut-edges} of the decomposition.
 \end{inparaenum}
\end{definition}
Note that the diameter constraint refers to strong diameter, in
particular, each induced subgraph $G[V_i]$ must be connected.

Algorithms for $(\eps,(\log n)/\eps)$-decompositions were developed
in many contexts (e.g., parallel
algorithms~\cite{awerbuch1992low,blelloch2014nearly,miller2013parallel}). An
implementation in the \congest-model can be derived from the algorithm in~\cite{elkin2017efficient} for constructing spanners.
Specifically, we get the following as a corollary from~\cite{elkin2017efficient}.
\begin{corollary}
  A $(\eps, O(\log n/\eps))$-decomposition can be computed in the
randomized  \congest-model in $O((\log n)/\eps)$ rounds with probability at least $1-1/\poly(n)$.
\end{corollary}

A nice feature of the algorithm based on random exponential shifts is
that at the end of the algorithm, there is a spanning BFS-like rooted
tree $T_i$ for each subset $V_i$ in the decomposition. Moreover, each
vertex $v\in V_i$ knows the center of $T_i$ as well as its parent in
$T_i$. In addition, every vertex knows which of the edges
incident to it are cut-edges.

The following definition captures the notion of connected witnesses
against a graph satisfying a
property.
 \begin{definition}[\cite{censor2016fast}]
  A graph property $\prop$ is \emph{non-disjointed} if for every witness
  $G'$ against $G\in \prop$, there exists an induced subgraph $G''$ of
  $G'$ that is connected such that $G''$ is also a witness against
  $G\in \prop$.
\end{definition}
\begin{comment}
  In fact, in Theorem~\ref{thm:sim}, we can relax the definition and
  exclude the requirement that $G''$ is an \emph{induced} subgraph of
  $G'$ and require only that $G''$ is a subgraph of $G'$.
\end{comment}

The main result of this section is formulated in the following
theorem. We refer to a distributed algorithm in which all vertices
accept iff $G\in\prop$ as a \emph{verifier} for $\prop$.
\begin{theorem}\label{thm:sim}
  Let $\prop$ be an edge-monotone non-disjointed graph property that can
  be verified in the \congest-model in $O(\diam(G))$ rounds.  Then
  there is an $\eps$-tester for $\prop$ in the randomized \congest-model
  with $O((\log n)/\eps)$ rounds.
\end{theorem}
\begin{proof}
  The algorithm tries to ``fix'' $G$ so that it satisfies $\prop$ by
  removing less than $\eps\cdot m $ edges.  The algorithm consists of
  two phases. In the first phase, an $(\eps',O((\log n)/\eps')$
  decomposition is computed in $O((\log n)/\eps')$ rounds, for
  $\eps'=\eps/2$.  The algorithm removes all the cut-edges of the
  decomposition. (There are at most $\eps\cdot m/2$ such edges.)  In the
  second phase, in each subgraph $G[V_i]$, an independent execution of the
  verifier algorithm for $\prop$ is executed. The number of rounds of the
  verifier in $G[V_i]$ is $O(\diam(G[V_i])=O((\log n)/\eps)$.

  We first prove completeness. Assume that $G\in \prop$.  Since $\prop$ is an
  edge-monotone property, the deletion of the cut-edges does not
  introduce a witness against $\prop$.  This implies that each induced
  subgraph $G[V_i]$ does not contain a witness against $\prop$, and hence
  the verifier do not reject, and every vertex accepts.

  We now prove soundness. If $G$ is $\eps$-far from $\prop$, then after the
  removal of the cut-edges (at most $\eps m/2$ edges) property $\prop$ is
  still not satisfied. Let $G'$ be a witness against the remaining
  graph satisfying $\prop$. Since property $\prop$ is non-disjointed, there
  exists a connected witness $G''$ in the remaining graph. This witness
  is contained in one of the subgraphs $G[V_i]$, and therefore, the
  verifier that is executed in $G[V_i]$ will reject, hence at least
  one vertex rejects, as required.
\end{proof}
We remark that if the round complexity of the verifier is
$f(\diam(G),n)$ (e.g., $f(\Delta,n)=\Delta+\log n$), then the round
complexity of the $\eps$-tester is $O((\log n)/\eps)+f((\log
n)/\eps,n)$. This follows directly from the proof.

\if\arxiv=0
\paragraph*{Extensions to $\eps$-Testers}
\else
\paragraph*{Extensions to $\eps$-Testers.}
\fi
The following ``bootstrapping'' technique can be applied.  If there
exists an $\eps$-tester in the \congest-model with round complexity
$O(\diam(G))$, then there exists an $\eps$-tester with round
complexity $O((\log n)/\eps)$. The proof is along the same lines,
expect that instead of a verifier, an $\eps'$-tester is executed in
each subgraph $G[V_i]$. Indeed, if $G$ is $\eps$-far from $\prop$,
then, by an averaging argument, there must exist a subset $V_i$ such that $G[V_i]$ is $\eps'$-far
from $\prop$. Otherwise, we could ``fix'' all the parts by deleting at
most $\eps' \cdot m$ edges, and thus ``fix'' $G$ by deleting at most
$2\eps'\cdot m =\eps m$ edges, a contradiction.

%%%%%%%%%%%%%%%%%%%%%%%%%%%%%%%%%%%%
\subsection{Testing Bipartiteness}
Theorem~\ref{thm:sim} can be used to test whether a graph is
bipartite or $\eps$-far from being bipartite.  A verifier for
bipartiteness can be obtained by attempting to $2$-color the vertices
(e.g., BFS that assigns alternating colors to layers). In our special
case, each subgraph $G[V_i]$ has a root which is the only vertex that
initiates the BFS. In the general case, one would need to deal with
``collisions'' between searches, and how one search ``kills'' the other
searches initiated by vertices of lower ID.

%%%%%%%%%%%%%%%%%%%%%%%%%%%%%%%%%%%%
\subsection{Testing Cycle-freeness}
Theorem~\ref{thm:sim} can be used to test whether a graph is
acyclic or $\eps$-far from being acyclic.  As in the case of
bipartiteness, any scan (e.g., DFS, BFS) can be applied. A second
visit to a vertex indicates a cycle, in which case the vertex rejects.

\begin{corollary}
  There exists an $\eps$-tester in the randomized \congest-model for
  bipartiteness and cycle-freeness with round complexity $O((\log n)/\eps)$.
\end{corollary}

%%%%%%%%%%%%%%%%%%%%%%%%%%%%%%%%%%%%
\subsection{Corrector for Cycle-Freeness}
Our $\eps$-testers for testing cycle freeness can be easily converted
into $\eps$-correctors by removing the following edges: (1)~All the
cut-edges are removed.  (2)~In each $G[V_i]$, all the edges which are
not in the BFS-like spanning tree $T_i$ are removed (in order to maintain
consistency with the same spanning tree one needs to define a
consistent way for breaking ties, e.g., by the ID of the vertices).

Therefore, the total number of edges that we keep is at most $|V| + \eps|E|$.

\begin{theorem}
  There exists an $\eps$-corrector for cycle-freeness in the randomized
  \congest-model with round complexity $O((\log n)/\eps)$.
\end{theorem}
\begin{comment}
  An $\eps$-corrector for bipartiteness removes the following edges:
  (1)~All the cut-edges are removed. (2)~Whenever a conflict occurs in
  the $2$-coloring algorithm (e.g., BFS that assigns alternating
  colors to layers), the edge along which the conflicting color was
  sent is removed.
\end{comment}

%%% Localip Variables:
%%% mode: latex
%%% TeX-master: "distTesting"
%%% End:

%%%%%%%%%%%%%%%%%%%%%%%%%%%%%%%%%%%%
\section{Testing $H$-freeness in $\Theta(1/\eps)$ Rounds for $|V(H)|\leq 4$}\label{sec:hfour}

\subsection{Testing Triangle-freeness}
In this section we present an $\eps$-tester for triangle-freeness that
works in the \congest-model. The number of rounds is $O(1/\eps)$.

Consider a violating edge $\{A, B\}$ and a corresponding  triangle $ABC$ in the graph $G=(V,E)$. This triangle can be
detected if $A$ tells $B$ about a neighbor $C\in N(A)$ with the hope
that $C$ is also a neighbor of $B$.  Vertex $B$ checks that $C$ is
also its neighbor, and if it is, then the triangle $ABC$ is
detected. Hence, $A$ would like to send to $B$ the name of a vertex
$C$ such that $C\in N(A)\cap N(B)$. Since $A$ can discover $N(A)$ in a
single round, it proceeds by telling $B$ about a neighbor $C\in
N(A)\setminus\{B\}$ chosen uniformly at random. Let $\msg{A}{B}$
denote the random neighbor that $A$ reports to $B$.  A listing of the
distributed $\eps$-tester for triangle-freeness appears as
Algorithm~\ref{alg:triangletest}. Note that all the messages
$\{\msg{A}{B}\}_{(A,B)\in E}$ are independent, and that the messages
are re-chosen for each iteration.

\begin{claim}\label{claim:single triangle}
Let $\{A, B\}$ be a violating edge, then $\Prr{\msg{A}{B}\in N(B)}\geq 1/m$.
\end{claim}
\begin{proof}
Let $ABC$ be a triangle which is a witness for the violation of $\{A, B\}$.
%The event that triangle $ABC$ is detected is contained in the event that $\msg{A}{B}\in N(B)$.
Since $ABC$ is a triangle, $C\in N(A)\cap N(B)$, and $\Prr{\msg{A}{B}\in N(B)}\geq 1/\deg(A) \geq 1/m$.
\end{proof}

%\begin{claim}\label{claim:cover3}
 % If a graph $G$ is $\eps$-far from being triangle-free, then it
%  contains at least $\eps\cdot m /3$ edge-disjoint triangles.
%\end{claim}
%\begin{proof}
%  Consider the following procedure for ``covering'' all the triangles:
%  while the graph contains a triangle, delete all three edges of the
 % triangle.  When the procedure ends, the remaining graph is
 % triangle-free, hence at least $\eps m$ edges were removed. The set of
 % deleted triangles is edge disjoint and hence contains at least $\eps
 % m/3$ triangles.
%\end{proof}

\begin{theorem}
	Algorithm~\ref{alg:triangletest} is an $\eps$-tester for triangle-freeness.
\end{theorem}
\begin{proof}
	\emph{Completeness:} If $G$ is triangle free then Line~\ref{line:rejtriangle} is never satisfied, hence for every $v$ Algorithm~\ref{alg:triangletest} terminates at Line~\ref{line:acctriangle}.
	
	\emph{Soundness:} Let $G=(V,E)$ be a graph which is $\eps$-far
        from being triangle free. Therefore there exist at least $\eps\cdot m$ edges, each belonging to at least one triangle. Hence, the probability of not detecting any
        of these triangles in a single iteration is at most
        $(1-1/m)^{\eps m}$. The reject probability is amplified to
        $2/3$ by setting the number of iterations to be $\Theta(1/\eps)$.
\end{proof}

\begin{algorithm}\small
\DontPrintSemicolon
%\KwIn{The $\gamma$-minimal set $S$ at time $t$ wrt to $x$.}
%\KwOut{Update of the primal and dual solutions $x$ and $y$ (global variables).}
%\KwResult{REJECT if $v$ is in a triangle, otherwise ACCEPT.}
Send $v$ to all $u\in N(v)$
\tcp*{1st round: each $v$ learns $N(v)$}
\For{$t= \Theta(1/\eps)$ times}{
For all $u \in N(v)$, simultaneously:
send $u$ the message $\msg{v}{u} \sim U(N(v)\setminus \{u\})$.\;
If $\exists w\in N(v)$ such that $\msg{w}{v}\in N(v)$ then \KwRet{REJECT}\nllabel{line:rejtriangle}\;
}
\KwRet{ACCEPT}\nllabel{line:acctriangle}\;
\caption{Triangle-free-test$(v)$\label{alg:triangletest}}
\end{algorithm}

\subsection{Testing $C_4$-freeness in $\Theta(1/\eps)$ Rounds}
In this section we present an $\eps$-tester  in the \congest-model for $C_4$-freeness that
runs in $O(1/\eps)$ rounds.

\if\arxiv=0
\paragraph*{Uniform Sampling of $2$-paths}
\else
\paragraph*{Uniform Sampling of $2$-paths.}
\fi
Let $P_2(v)$ denote the set of all paths of length $2$ that start at
$v$.  The algorithm is based on the ability of each vertex $v$ to
uniformly sample a path from $P_2(v)$.  How many paths in $P_2(v)$
start with the edge $(v,w)$? Clearly, there are $(\deg(w)-1)$ such
paths.  Hence the first edge should be chosen according to the degree
distribution over $N(v)$ defined by $\pi^v(w)\triangleq
(\deg(w)-1)/\sum_{x\in N(v)} (\deg(x)-1)$.  Moreover, for each $x\in
N(w)\setminus \{v\}$, the (directed) edge $(w,x)$ appears exactly once as the
second edge of a path in $P_2(v)$.  Hence, given the first edge, the
second edge is chosen uniformly.

This implies that $v$ can pick a random path $p\in P_2(v)$ as follows:
(1)~Each neighbor $w\in N(v)$ sends $v$ a uniformly randomly chosen
neighbor $B_v(w)\in N(w)\setminus \{v\}$.  The edge $(w,B_v(w))$ is a
candidate edge for the second edge of $p$.  (2)~$v$ picks a neighbor
$A(v)\in N(v)$ where $A(v) \sim \pi^v$. The random path $p$ is
$p=\langle v,A(v),B_v(A(v)) \rangle$, and it is uniformly distributed over $P_2(v)$.

In the algorithm, vertex $v$ reports a path to each neighbor. We
denote by $p_u(v)$ the path in $P_2(v)$ that $v$ reports to $u\in
N(v)$. This is done by independently picking neighbors $A_u(v)\in
N(v)$, where each $A_u(v)\sim \pi^v$. Hence, the path that $v$ reports
to $u$ is $p_u(v) \triangleq \langle v,A_u(v),B_v(A_u(v)) \rangle$
Algorithm~\ref{alg:htest} uses this process for reporting paths of
length $2$. Interestingly, these paths are not independent, however
for the case of edge disjoint copies of $C_4$, their ``usefulness'' in
detecting copies of $C_4$ turns out to be independent (see
Lemma~\ref{lemma:ind}).

\if\arxiv=0
\paragraph*{Detecting a Cycle}
\else
\paragraph*{Detecting a Cycle.}
\fi
Consider a copy $C=(v,w,x,u)$ of $C_4$ in $G$. If the $2$-path $p_u(v)$ that $v$
reports to $u$ is $p_u(v)=(v,w,x)$, then $u$ can check whether the
last vertex $x$ in $p_u(v)$ is also in $N(u)$. If $x\in N(u)$, then the
copy $C$ in $G$ of $C_4$ is detected. (The vertex $u$ also needs to verify that $w\neq u$.)

\if\arxiv=0
\paragraph*{Description of the Algorithm}
\else
\paragraph*{Description of the Algorithm.}
\fi
The $\eps$-tester for $C_4$-freeness is listed as
Algorithm~\ref{alg:htest}.  In the first round, each vertex $v$ learns
its neighborhood $N(v)$ and the degree of each neighbor.  The for-loop
repeats $t=O(1/\eps)$ times. Each iteration consists of three
rounds. In the first round, $v$ independently draws fresh values for
$A_u(v)$ and $B_u(v)$ for each of its neighbors $u\in N(v)$, and sends
$B_u(v)$ to $u$.  In the second round, for each neighbor $u\in N(v)$,
$v$ sends the path $\langle v,A_u(v), B_v(A_u(v))\rangle$.  In the
third round, $v$ checks if it received a path $\langle w,a,b\rangle$
for a neighbor $w\in N(v)$ where $a\neq v$ and $b\in N(v)$. If this
occurs, then $(v,w,a,b)$ is a copy of $C_4$, and vertex $v$ rejects.
If $v$ did not reject in all the iterations, then it finally accepts.

\if\arxiv=0
\paragraph*{Analysis of the Algorithm}
\else
\paragraph*{Analysis of the Algorithm.}
\fi
\begin{definition}
We say that $p_u(v)$ is a \emph{success} (wrt $C=(v,w,x,u)$) if $p_u(v)=(v,w,x)$.
Let $I_{v,u}$ denote the indicator variable of the event that $p_u(v)$ is a success.
\end{definition}

\begin{lemma}\label{lemma:ind}
Let $\{C^j(v_j,w_j,x_j,u_j)\}_{j\in J}$ denote a set of edge-disjoint copies of $C_4$ in $G$.
Then the random variables $I_{v_j,u_j}$ are independent.
\end{lemma}
\begin{proof}
  The event $I_{v,u}=1$ occurs iff $A_u(v)=w$ and $B_v(w)=x$. Both
  $A_u(v)$ and $B_v(w)$ are random variables assigned to (directed)
  edges. By construction, all the random variables $\{A_u(v)\}_{(u,v\in E} \cup\{
  B_v(w)\}_{(v,w)\in E}$ are independent. Since the cycles are
  edge-disjoint, the lemma follows.
\end{proof}

\begin{claim}\label{claim:single C}
  $\Pr[I_{v,u}=1 \mid C] \geq 1/(2m)$.
\end{claim}
\begin{proof}
  The path $p_u(v)$ equals $(v,w,x)$ iff $A_u(v)=w$ and $B_v(w)=x$. As
  $A_u(v)$ and $B_v(w)$ are independent, we obtain
  \begin{align*}
    \Prr{I_{v,u}=1 \mid C}&=\Prr{A_u(v)=w|C} \cdot \Prr{B_v(w)=x\mid C}\\
&=\frac{\deg(w)-1}{\sum_{x\in N(v)} (\deg(x)-1)} \cdot \frac {1}{\deg(w)-1}\geq \frac{1}{2m}.
  \end{align*}
\end{proof}

\begin{claim}\label{claim:cover}
  If a graph $G$ is $\eps$-far from being $C_4$-free, then it
  contains at least $\eps\cdot m /4$ edge-disjoint copies of $C_4$.
\end{claim}
\begin{proof}
  Consider the following procedure for ``covering'' all the copies of $C_4$:
  while the graph contains a copy of $C_4$, delete all four edges of the
  copy.  When the procedure ends, the remaining graph is
  $C_4$-free, hence at least $\eps m$ edges were removed. The set of
  deleted copies of $C_4$ is edge disjoint and hence contains at least $\eps
  m/4$ copies of $C_4$.
\end{proof}

\begin{theorem}
  Algorithm~\ref{alg:htest} is an $\eps$-tester for
  $C_4$-freeness. The round complexity of the algorithm is
  $\Theta(1/\eps)$ and in each round no more than $O(\log n)$ bits are
  communicated along each edge.
\end{theorem}
\begin{proof}
  \emph{Completeness:} If $G$ is $C_4$-free then Line~\ref{line:rej}
  is never satisfied, hence for every $v$ Algorithm~\ref{alg:htest}
  terminates at Line~\ref{line:acc}.
	
  \emph{Soundness:} Let $G=(V,E)$ be a graph which is $\eps$-far from
  being $C_4$-free. Therefore, there exist $\ell\triangleq \eps m/4$
  edge disjoint copies of $C_4$ in $G$. Denote these copies by
  $\{C^1,\ldots, C^{\ell}\}$, where $C^j=(v_j,w_j,w_j,u_j)$. In each
  iteration, the cycle $C^j$ is detected if $I_{v_j,u_j}=1$, which (by
    Claim~\ref{claim:single C}) occurs with probability at least
    $1/(2m)$. The cycles $\{C^j\}_j$ are edge-disjoint, hence, by
    Lemma~\ref{lemma:ind}, the probability that none of these cycles
    is detected is at most $(1-1/(2m))^\ell$. The iterations are
    independent, and hence the probability that all the iterations
    fail to detect one of these cycles is at most
    $(1-1/(2m))^{\ell\cdot t}$. Since $\ell=\eps m/4$, setting
    $t=16/\eps$ reduces the probability of false accept to at most $1/3$, as required.
\end{proof}

\begin{algorithm}\small
\DontPrintSemicolon
%\KwIn{A graph $H$ such that $V(H)= 4$, and $P_4 \subseteq H$.}
%\KwOut{Update of the primal and dual solutions $x$ and $y$ (global variables).}
%\KwResult{REJECT if $v$ is in a copy of $P_4$, otherwise ACCEPT.}
%\Begin{
Send $v$ and $d(v)$ to all $u\in N(v)$
\tcp*{$v$ learns $N(v)$ and $d(u)$ for every $u\in N(v)$}

Define the following distribution $\pi^v$ over $N(v)$: For every $w \in N(v)$,  $\pi^v(w) \triangleq d(w)/\sum_{x \in N(v)}d(x)$ .\;

\For{$t\triangleq (16/\eps)$ times}{
For every neighbor $u\in N(v)$ independently draw $A_u(v)\sim\pi^v$ and $B_u(v)\sim U(N(v)\setminus \{u\}$, send $B_u(v)$ to $u$.
\;

For every neighbor $u\in N(v)$ send the path $\langle v,A_u(v),B_v(A_u(v)) \rangle$ to $u$.\label{line:msg}
\;

\If{$\exists w\in N(v)$ s.t. $v$ received the path $\langle w,a,b\rangle $ from
  $w$, where  $v \neq a$ and $b\in N(v)$}{
  \KwRet{REJECT}\nllabel{line:rej}\tcp*{A cycle $C=(v, w, a, b)$ was
    found.}}  }
\KwRet{ACCEPT}\nllabel{line:acc}

\caption{$C_4$-free-test$(v)$\label{alg:htest}}
\end{algorithm}

 \begin{comment}
%   <llx> <lly> <urx> <ury>
\begin{figure}[t!]
  \centering
  \includegraphics[width=0.45\textwidth,trim=0 80bp 0 0bp,clip]{c4}
  \caption{\moti{caption}Example for the $H$-free testing for $H$ of
    order $4$.}
  \label{fig:toyp4}
\end{figure}
 \end{comment}

\if\arxiv=0
\paragraph*{Extending Algorithm~\ref{alg:htest}}
\else
\paragraph*{Extending Algorithm~\ref{alg:htest}.}
\fi
The algorithm can be easily extended to test $H$-freeness for any connected $H$ over four nodes. If $H$ is a $K_{1,3}$ then clearly $H$-freeness can be tested in one round.
Otherwise, $H$ is Hamiltonian and can be tested by simply sending an additional bit in the message sent in Line~\ref{line:msg} of the algorithm.
The additional bit indicates whether $v$ is connected to $B_v(A_u(v))$. Given this information, $u$ can determine the subgraph induced on $\{u, v, A_u(v), B_v(A_u(v))\}$, and hence rejects if $H$ is a subgraph of this induced subgraph.
Therefore we obtain the following theorem.

\begin{theorem}
  There is an algorithm which is an $\eps$-tester for
  $H$-freeness for any connected $H$ over $4$ vertices. The round complexity of the algorithm is
  $\Theta(1/\eps)$ and in each round no more than $O(\log n)$ bits are
  communicated along each edge.
\end{theorem}

%%% Local Variables:
%%% mode: latex
%%% TeX-master: "distTesting"
%%% End:

%%%%%%%%%%%%%%%%%%%%%%%%%%%%%%%%%%%%
\section{Testing $T$-freeness for any tree $T$}\label{sec:bigtrees}
In this section we first generalize the tester by Iwama and
Yoshida~\cite{IwamaY14} of testing $k$-path freeness to testing the
exclusion of any tree, $T$, of order $k$. Our tester has a one sided
error and it works in the general graph model with random edge
queries. We, then, show how to simulate this algorithm in the \congest-
model in $O(k^{k^2+1}\cdot\eps^{-k})$ rounds.  We assume that the
vertices of $T$ are labeled by $v_0, \ldots, v_{k-1}$.
%Note that for $|V(T)| \leq 4$ we described in previous sections how to obtain an $\eps$-tester that operates in $O(1/\eps)$ rounds.

\subsection{Global Algorithm Description and Analysis}

\paragraph*{Global Algorithm Description.}
The algorithm by Iwama and Yoshida~\cite{IwamaY14} for testing
$k$-path freeness proceeds as follows. An edge is picked u.a.r. and an
endpoint, $v$, of the selected edge,is picked u.a.r.  A random walk of
length $k$ is performed from $v$, if a simple path of length $k$ is
found then the algorithm rejects. The analysis in~\cite{IwamaY14}
shows that this process has a constant probability (depends only on
$k$ and $\eps$) to find a $k$-path in an $\eps$-far from $k$-path
freeness graph.

We generalize this tester in the following straightforward manner. We
pick a random vertex $v$ as in the above-mentioned algorithm.  The
vertex $v$ is a candidate for being the root of a copy of $T$. For the
sake of brevity we denote the (possible) root of the copy of $T$ also
by $v_0$.
%The idea is rather simple. Pick a random vertex $v_0$ (or more precisely a random edge and a random endpoint of this edge in case the graph is not connected) and
From $v$ we start a ``DFS-like'' revealing of a tree which is a possible copy of $T$ with the first random vertex acting as its root.
DFS-like means that we scan a subgraph of $G$ starting from $v$ as follows: the algorithm independently and randomly selects $d_T(v_0)$ neighbors (out of the possible $d_G(v)$) and recursively scans the graph from each of these randomly chosen neighbors. While scanning, if we encounter any vertex more than once then we abort the process (we did not find a copy of $T$). If the process terminates, then this implies that the algorithm found a copy of $T$. In order to obtain probability of success of $2/3$ the above process is repeated $t=f(\eps, k)$ times. The listing of this algorithm appears in Algorithm~\ref{alg:centtreetest}.

\begin{algorithm}\small
\DontPrintSemicolon
%\KwIn{A graph $H$ such that $V(H)= 4$, and $P_4 \subseteq H$.}
%\KwOut{Update of the primal and dual solutions $x$ and $y$ (global variables).}
\For{$t\triangleq \Theta(k^{k^2}/\eps^k)$ times}{
Pick an edge u.a.r. and an endpoint, $v$, of the selected edge u.a.r.\;\label{step:trees1}
Initialize all the vertices in $G$ to be un-labeled.\;
Call Recursive-tree-exclusion$(T, 0, v)$ and \KwRet{REJECT} if it returned $1$.\;
}
\KwRet{ACCEPT}.
\caption{Global-tree-free-test$(T,v)$\label{alg:centtreetest}}
\end{algorithm}

\begin{procedure}
\DontPrintSemicolon
If $v$ was already labeled then return $0$, otherwise, label $v$ by $i$.\tcp*{The recursion returns $0$ if the revealed labeled  subgraph is not $T$.}
Define $\ell = d_T(v_i) -1$ if $i>0$ and $\ell = d_T(v_i)$ otherwise.\;
Let $v_{i_1}, \ldots, v_{i_\ell}$ denote the labels of the children of $v_i$ in $T$ (in which $v_0$ is the root).\;
Pick u.a.r. $\ell$ vertices $u_1, \ldots, u_\ell$ from $N_G(v)$ and recursively call Recursive-tree-exclusion$(T, i_j, u_j)$ for each $j \in [\ell]$\;\label{step:trees2}
If one of the calls returned $0$, then return $0$, otherwise return $1$.
\caption{Recursive-tree-exclusion($T,i,v$)}
\label{proc:treeexc}
\end{procedure}

\if\arxiv=0
\subparagraph*{Global Algorithm Analysis.}
\else
\paragraph*{Global Algorithm Analysis.}
\fi
\begin{theorem}\label{thm:tfree}
	Algorithm~\ref{alg:centtreetest} is a global $\eps$-tester, one-sided error for $T$-freenes. The query complexity of the algorithm is $O\left(k^{k^2+1} \cdot \eps^{-k}\right)$. The algorithm works in the general graph model augmented with random edge samples.
\end{theorem}
\begin{proof}
We closely follow the analysis approach by Iwama and Yoshida~\cite{IwamaY14}. If the input graph $G$ does not contain a copy of $T$ then the $\eps$-tester will not ``find'' any copy, and will return ACCEPT.

Let $G$ be $\eps$-far from being $T$-free.
This implies that there are at least $\ell = \eps m/(k-1)$ edge disjoint copies of $T$ (otherwise contradicting the assumption that $G$ is $\eps$-far being $T$-free). Let $\alpha= \{T_1, \ldots T_\ell\}$ denote such set of edge disjoint copies of $T$.
Note that two trees in $\alpha$ might intersect in one of their vertices.
%We  assume, for the sake of analysis, that the vertices of the input tree are colored by $k$ \emph{distinct} colors $\{1, \ldots, k\}$.
Recall that $T$ is labeled by $v_0,\ldots,v_{k-1}$.
For the sake of the analysis we consider a labeling of the vertices of $G$ where each vertex has a label in $\{v_j \mid j\in[k-1]\}$.
We say that a copy of $T$ in $\alpha$ is {\em labeled correctly} if the labeled $T$ and its labeled copy are the same (note that this is a more restrictive requirement then asking for an isomorphic mapping between the two objects).
We consider a subset of $\alpha$ which are the copies of $T$ which are labeled correctly, which we denote by $\alpha^c$.
Note that from linearity of expectation there exists a labeling such that the size of $\alpha^c$ is at least $\eps m / \left((k-1)\cdot k^k\right)$ - we fix our labeling to this one. Define $\beta \triangleq \eps / \left((k-1)\cdot k^k\right)$.

We now ``sparsify'' our set of trees even more.
We proceed in iterations, initially $\alpha_0^c = \alpha^c$ and define $\gamma \triangleq \beta/4$.
We continue to the $i$-th iteration as long as there exists a vertex, $y_i$, in $\alpha_{i-1}^c$ such that the number of copies of $T$ in $\alpha_{i-1}^c$ that $y_i$ belongs to, is less than $\gamma \cdot d_G(y)$.
In this case we obtain $\alpha_i^c$ from $\alpha_{i-1}^c$ by removing from $\alpha_{i-1}^c$ all the trees that contain $y_i$.
%Roughly speaking, this phase of removal guarantees that from each remaining vertex we can proceed with the DFS-like algorithm and pick a copy of some color with probability of at least $\gamma^{-1}$ (since we have $\gamma \cdot d_G(v)$ copies of each colored neighbor of $v$ for each $v$).
Note that this results in a deletion of at most $2\gamma \cdot |E(G)|$ copies of $T$ from $\alpha^c$, as each vertex gets removed at most once.
Therefore, at the last iteration, we obtain in this way the set $\alpha^*$ which contains at least $(\beta/2)\cdot m$, edge-disjoint, correctly colored, copies of $T$.

We now focus on the copies of $T$ in $\alpha^*$.
Fix an iteration of Algorithm~\ref{alg:centtreetest}, and consider the event that on Line~\ref{step:trees1}, the starting vertex, $v$, is a vertex in $\alpha^*$ which is labeled by $v_0$.
The probability that this event occurs is at least $\beta/4$.
This follows from the fact that there are at least $(\beta/2)\cdot m$ edges in $\alpha^*$ in which one of their endpoints is labeled by $v_0$.
Conditioned on this event, the probability that $u_1, \ldots, u_\ell$, that are chosen in Line~\ref{step:trees2} of Procedure~\ref{proc:treeexc}, are in $\alpha^*$ and their labels are $v_{i_1}, \ldots, v_{i_\ell}$, respectively, is at least $\gamma^\ell$.
Therefore, by an inductive argument, the probability that the algorithm finds $T$ is at least $(\beta/4)\cdot \gamma^{k-1} = \Omega(\eps^k/(k^{k^2}))$, for each iteration.
Since the number of iterations is $\Theta(k^{k^2}/\eps^k)$, the algorithm rejects $G$ with probability at least $2/3$.
%Let us recap: we now hold a set of copies of $T$ in $G$ which are correctly colored, they are edge disjoint, and each vertex is ``heavy''.
%Now we focus on this $\alpha_I^c(\gamma)$ set and bound the probability that the algorithm succeeds in revealing a tree in which is a copy $T$  (which is not necessarily in $\alpha_I^c(\gamma)$) in a given iteration - it is simply the probability that the algorithm finds the neighboring vertices which are colored like $T$ in each vertex given the success of randomly picking of a root in $\alpha_I^c(\gamma)$. This results in a success probability of at least $\gamma^{-k}$. For an appropriate selection of $\gamma$ and given the probability of ``hitting'' the root of a copy of $T$ in $\alpha_I^c(\gamma)$ the resulting probability is independent of $|V(G)|$. The resulting tester is a one sided error tester which works in the general graph model with random edge samples, as required. The rigorous analysis is presented in the following theorem.\footnote{The random edge queries are required for the case where $G$ is not a connected graph.  Otherwise, random vertex queries suffice.}
\end{proof}

\subsection{Simulating the Global tester in the \congest-\! model}
In order to simulate the above global tester in the \congest-model, we
first need to show how to simulate a random choice of the initial
vertex. Note, that in the distributed setting we are interested in the
case where $|E(G)| \geq |V(G)|-1$, hence in Line~\ref{step:trees2} of
Algorithm~\ref{alg:centtreetest}, one can simply pick a vertex uniformly at random.

% For the simulation, we can start with a vertex, $v$, which is
% picked u.a.r. and simulate Algorithm~\ref{alg:centtreetest} from $v$
% (since the graph is connected, $m \geq n-1$, and so in
% Line~\ref{step:trees2} of Algorithm~\ref{alg:centtreetest}, one can
% simply pick a vertex u.a.r.).

\paragraph*{Distributed \congest\ Simulation.}
%\moti{fill in the trivial sim.} \moti{can we do better?}
The simulation proceeds in three stages, as follows.

\emph{Random ranking assignment.} First, each vertex $v$ picks a random rank (independently), denoted by $r(v)$, to be a uniform random number in $[n^2]$.
%Now we only need to simulate the execution that is originated from the vertex with the highest rank.
%The distributed algorithm proceeds as follows.

\emph{Labeling ``competition'' phase.} Let us call the vertex $v_0$ the vertex at level $1$, its children the vertices at level $2$ and so on.
Now, at the first round: each vertex label itself as $v_0$, picks random neighbors as in Line~\ref{step:trees2} of Algorithm~\ref{alg:centtreetest} and sends to each selected vertex its respective label.
At the next round, each vertex $v$ picks the label that was assigned to it by the vertex with the highest rank, which we refer to as the root.
If $v$'s label is of  level $2$ then it continues by simulating Procedure~\ref{proc:treeexc}. Specifically, it selects random neighbors and their respective labels and send to each selected neighbor a message with the respective label and the rank of its root.
And so this process continues for $k$ rounds.

\emph{Labeling verification phase.} Now, one needs to report to each root if it succeeded in finding $T$.
This can be achieved in additional $k$ rounds: in rounds $i$ all the vertices that were assigned with label from level $i$ send their parent in the revealed tree, whether they succeeded or not, that is,  a leaf reports on success if it was not labeled more than once by the respective root, and an internal node reports on success if all its children report on success and if it was not labeled more than once by the respective root.
The correctness of the simulations follows from the fact that the messages that are originated from the vertex with the highest rank always get prioritizes and thus Algorithm~\ref{alg:centtreetest} is simulated.

\begin{theorem}
	There is an $\eps$-tester in the \congest-model that on input $T$, where $T$ is a tree, the tester tests if the graph is $T$-free. The rounds complexity of the tester is $O\left(k^{k^2+1} \cdot \eps^{-k}\right)$ where $k$ is the order of $T$.
\end{theorem} 
%%%%%%%%%%%%%%%%%%%%%%%%%%%%%%%%%%%%

%\subsection*{Discussion and Future Work}
%\moti{Testing H freenes for k>4 for H that includes Pk.}
%\moti{Testing $H$-freeness in Minor-free Graphs?}\moti{Testing Planarity?}
%\moti{Semi-Adaptive (Centralized) Testers to Distributed Testers? random-walk based testers can be efficiently simulated such as the Pk freenes. Sociology remark, and the algorithm by Goldreich-Ron}
%\moti{We (maybe) got cloaser to refuting Pierres conjecture from DISC 2016 that for k > 4 algorithms should have dependency on n}\moti{connections to Centralized Testing and carrying lower bounds from there to here.}

%\cite{*}
%\bibliographystyle{alpha}%{apalike}
\bibliography{test}
\end{document}